\documentclass[11pt]{article} 
\usepackage[dvipsnames,svgnames,x11names]{xcolor}
\usepackage{amsmath, amsthm, thm-restate}
\usepackage{xspace}
\usepackage{float}
\usepackage[T1]{fontenc}
\usepackage{mdframed}
\usepackage{mathtools}
\usepackage{enumitem}

\usepackage[linesnumbered,ruled]{algorithm2e}

\usepackage{algpseudocode}
\usepackage{authblk}

\usepackage{fullpage}
\usepackage{thmtools}
\usepackage[colorlinks,hypertexnames=false]{hyperref}
\hypersetup{colorlinks={true},linkcolor={blue},citecolor=magenta}

\usepackage{amssymb,amsfonts,amscd,dsfont,mathrsfs}
\usepackage{complexity}
\usepackage{multirow}
\usepackage{mathpazo}
\usepackage{braket}
\usepackage{comment}
\usepackage{quantikz}
\usepackage[
backend=biber,
style=alphabetic,
sorting=ynt,
maxnames=99
]{biblatex}
\usepackage{url}
\usepackage{bbm}
\usepackage[capitalize]{cleveref}
\usepackage{hyperref} 
\usepackage[margin=1in]{geometry}
\hypersetup{breaklinks=true}

\theoremstyle{plain}
\newtheorem{theorem}{Theorem}[section]
\newtheorem{lemma}[theorem]{Lemma}
\newtheorem*{claim*}{Claim}

\newtheorem{claim}[theorem]{Claim}

\newtheorem*{theorem*}{Theorem}
\newtheorem*{definition*}{Definition}

\newtheoremstyle{boxed}{-\topsep}{}{\itshape}{}{\bfseries}{.}{.5em}{}
\theoremstyle{boxed}

\newmdtheoremenv[backgroundcolor=gray!10,
                 linewidth=0pt,
                 innerleftmargin=4pt,
                 innerrightmargin=4pt,
                 innertopmargin=4pt,
                 innerbottommargin=4pt,
            splitbottomskip=4pt,
            nobreak=true]{problem}[prob]{Problem}

\newmdtheoremenv[backgroundcolor=gray!10,
                 linewidth=0pt,
                 innerleftmargin=4pt,
                 innerrightmargin=4pt,
                 innertopmargin=1pt,
                 innerbottommargin=5.5pt,
            splitbottomskip=4pt]{conjecture}[conj]{Conjecture}

\let\originalleft\left
\let\originalright\right
\renewcommand{\left}{\mathopen{}\mathclose\bgroup\originalleft}
\renewcommand{\right}{\aftergroup\egroup\originalright}

\mathchardef\mhyphen="2D

\newcommand{\norm}[1]{\left\Vert {#1} \right\Vert}

\makeatletter
\renewcommand\AB@authnote[1]{\rlap{\textsuperscript{\normalfont#1}}}

\makeatother

\title{On State Distinguishing Without Inverse Queries}
\author[1]{Emma Wang}
\author[2]{Joseph Carolan}
\author[2]{Andrew M. Childs}

\affil[1]{White Station High School}
\affil[2]{Department of Computer Science, Institute for Advanced Computer Studies, and Joint Center for Quantum Information and Computer Science, University of Maryland}

\date{}

\begin{document}

\maketitle

\begin{abstract}
In the quantum state distinguishing problem, an algorithm is given copies of a state $\ket{\psi}$ or $\ket{\phi}$ separated by trace distance $\varepsilon$, and tasked with deciding which is the case. It is a standard result that $\Theta(\varepsilon^{-2})$ copies are required, while having access to a unitary state preparation oracle and its inverse improves this to $\Theta(\varepsilon^{-1})$ queries. We study the intermediate case with forward-only access, capturing scenarios where the oracle is a physical process that cannot be reversed. Tang and Wright recently showed that distinguishing a known state from an \emph{unknown} close state requires $\Omega(\varepsilon^{-2})$ forward queries in high dimension; we strengthen this result by proving the same bound even when both states are known, as in the state distinguishing setting. Moreover, our lower bound uses a direct hybrid argument, avoiding advanced compressed oracle techniques used in the prior work.
In contrast, we give a dimension-independent algorithm solving a continuous-time variant of state distinguishing with only $O(\varepsilon^{-1})$ total forward evolution time. The algorithm uses quantum Zeno dynamics to confine evolution to a known two-dimensional subspace. Thus, a quadratic advantage survives without inverse access for a natural continuous-time analogue of state distinguishing.
\end{abstract}

\section{Introduction}

On a digital quantum computer, implementing a unitary $U$ described by an explicit quantum circuit is no easier nor harder than implementing the inverse unitary $U^\dagger$. However, when the unitary $U$ represents a physical process or laboratory experiment, then it can be much harder to implement $U^\dagger$, in some cases as hard as reversing time within the system~\cite{TangWright25}. In the generic case, $\Omega(2^{2n})$ uses of an $n$-qubit black-box unitary are needed to implement its inverse deterministically~\cite{odake} or even approximately~\cite{CYZ26}, so this task is intractable in large dimensions. This motivates the study of query algorithms that only have forward access to an unknown unitary oracle.

Examples that do not satisfy this criterion are amplitude amplification and amplitude estimation, which generalize Grover's quadratic speedup. These algorithms involve repeatedly alternating a state-preparation procedure with its inverse~\cite{Grover96,BHMT02} to amplify an ``accepting'' component with amplitude $\varepsilon$. Tang and Wright recently proved that this is necessary, by exhibiting an instance for which an oracle \(U\) alone requires \(\Omega(\varepsilon^{-2})\) queries in large dimension, while access to both \(U\) and \(U^\dagger\) permits the superior \(O(\varepsilon^{-1})\) scaling~\cite{TangWright25}.

In particular, the amplitude amplification problem is formalized in~\cite{TangWright25} by asking an algorithm to amplify an unknown state $\ket{\phi_{\operatorname{good}}}$, given oracle access to a unitary that prepares a state with some amplitude on $\ket{\phi_{\operatorname{good}}}$ tensored with a success flag, and the remaining amplitude on a junk state $\ket{\phi_{\operatorname{bad}}}$ tensored with a failure flag. A succeeding algorithm must distinguish such an oracle from an oracle that has no amplitude on the success branch, which the authors use to show lower bounds. The algorithm necessarily does not know the description of $\ket{\phi_{\operatorname{good}}}$.

In this work we study whether this barrier persists in the \emph{state distinguishing} setting, where we are trying to distinguish unitaries preparing one of two \emph{known} states. This contrasts with the above formalization of amplitude amplification, which can be viewed as distinguishing an \emph{unknown} state with small amplitude on a success flag from an \emph{unknown} state with no amplitude on a success flag. 
We study two related settings:
\begin{enumerate}
    \item The \emph{discrete-time} setting, where the unknown oracle is a unitary operation $U$ preparing one of two known states.
    \item The \emph{continuous-time} setting, where the unknown operation is generated by a Hamiltonian $H$ whose first column is promised to be one of two possibilities, and arbitrarily short \emph{forward} evolutions are available.
\end{enumerate}

In the first case, we show that state distinguishing is hard, in particular requiring $\Theta(\varepsilon^{-2})$ queries in sufficiently high dimension, strengthening the main result of~\cite{TangWright25}. Moreover our result relies only on a basic hybrid argument, avoiding the need for compressed oracles and other advanced machinery used in the prior work.

On the other hand, we identify a continuous-time version of state distinguishing can be solved using dimension-independent $O(\varepsilon^{-1})$ total evolution time. This provides a natural Hamiltonian analog of the problem studied in~\cite{TangWright25} in which the Grover-type quantum advantage persists even without inverse access. Note that our formulation of continuous-time state distinguishing involves a promise that the action of the Hamiltonian on $\ket{0}$ is one of two known values, i.e., a promise on the first column of the Hamiltonian. This is not the same as the ability to apply arbitrary fractional powers of the discrete unitary in the discrete query problem, as such fractional powers need not satisfy this promise.

\subsection{Technical Overview}

In more detail, we define the discrete-time state distinguishing problem as follows. Following Tang and Wright~\cite{TangWright25}, we consider the high-dimensional setting where $d$ is super-polynomial in $\varepsilon^{-1}$. This setting captures laboratory experiments with many degrees of freedom.

\begin{problem}[Discrete-time State Distinguishing]
Let $\varepsilon \in (0,1]$ and $U\in\mathbb{C}^{d\times d}$ for $d \ge 2$ be an unknown unitary, given through forward-only query access, promised that exactly one of the following holds:
\begin{align*}
    \textnormal{\textbf{YES:}} \qquad & U\ket{0}=\ket{0},&
    \textnormal{\textbf{NO:}} \qquad & U \ket{0} = \sqrt{1-\varepsilon^2} \ket{0} + \varepsilon \ket{1}.
\end{align*}
The goal is to determine which case holds, given knowledge of $\varepsilon$.
\label{prob:disc-sd}
\end{problem}

Note that this problem is potentially easier than the full amplitude amplification task studied in \cite{TangWright25}, in that the \textbf{NO} case prepares a fixed target state. While we do not explicitly consider controlled access, all our lower bounds apply in the case where it is available as well.

Our first result is a proof, using a direct hybrid argument, that this problem requires $\Omega(\varepsilon^{-2})$ queries. In particular, we analyze distinguishing a random in-place permutation on $\ket{1}, \dots, \ket{d-1}$ from the same permutation followed by a small rotation on the span of $\ket{0}, \ket{1}$. It is known that inverting an in-place permutation is hard, which implies that the algorithm cannot put non-trivial amplitude on $\ket{\pi^{-1}(1)}$. Given this fact, we show that these two oracles can increase the distance of two $\delta$-separated states by at most $\sqrt{\delta^2+O(\varepsilon^2)}$ per query, which in turn shows that $\Omega(\varepsilon^{-2})$ queries are required. A matching upper bound follows from, for instance, using $U$ to directly prepare copies of the two states and running standard state discrimination. This gives the following result, proven in \cref{sec:lower}.

\begin{theorem}
    \Cref{prob:disc-sd} requires $\Theta(\varepsilon^{-2})$ queries in sufficiently high dimension ($d=\Omega(\varepsilon^{-8}$)).
    \label{thm:main-lower}
\end{theorem}

This naturally leaves the question of whether a Grover-type speed-up applies for any natural weakening of this problem. Towards this, we formulate a continuous-time version of the problem as follows.

\begin{problem}[Continuous-time State Distinguishing]
Let $H\in\mathbb{C}^{d\times d}$ for $d \ge 2$ be an unknown Hamiltonian with $\norm{H}\leq 1$, given through forward-only query access to $e^{-iHt}$ at cost $t$ per application, promised that exactly one of the following holds:
\begin{align*}
    \textnormal{\textbf{YES:}} \qquad & H\ket{0}=0,&
    \textnormal{\textbf{NO:}} \qquad & H\ket{0} = \varepsilon \ket{1}.
\end{align*}
The goal is to determine which case holds, given knowledge of $\varepsilon \in (0,1]$.
\label{prob:cont-sd}
\end{problem}

Note that this is not precisely the same as providing Hamiltonian query access to the unitary $U$ from \cref{prob:disc-sd}; in particular, $H$ satisfying the promise of \cref{prob:cont-sd} does not imply that $e^{-iH}$ satisfies the promise of \cref{prob:disc-sd}. However, such a problem would be unnatural in the Hamiltonian picture, as it corresponds to a promise about a column of a matrix exponential of $H$, rather than a column of $H$. In contrast, our formulation captures a natural analog in which we are promised structure on the first column of $H$ itself, mirroring the promise on the stucture of the first column of $U$ in \cref{prob:disc-sd}. We therefore view this version as a natural continuous analog of the discrete state distinguishing problem.

Unrestricted evolution under \(H\) leads to a loss of coherence in the non-$\ket{0}$ subspace, and therefore the time to distinguish the states can scale quadratically. Nonetheless, inverse evolution is unnecessary: by interspersing short forward evolutions with projections $\Pi$ onto \(\mathrm{span}\{\ket{0},\ket{1}\}\), one obtains quantum Zeno dynamics generated by the projected Hamiltonian \(\Pi H\Pi\)~\cite{MisraSudarshan77,ExnerIchinose05,FacchiPascazio08}. The projections can be implemented through projective measurements, reducing the problem to a two-dimensional evolution and yielding an \(O(\varepsilon^{-1})\)-time distinguishing algorithm using only forward evolution. This demonstrates a qualitative separation between forward-only access to a fixed unitary and forward-only access to continuous-time evolution, through the following theorem, proven in \cref{sec:algorithm,sec:analysis}.

\begin{theorem}
    \Cref{prob:cont-sd} requires $\Theta(\varepsilon^{-1})$ total evolution time.
    \label{thm:main-upper}
\end{theorem}
\subsection{Related Work}
\label{sec:related-work}

\paragraph{State distinguishing and inverse queries.}
Quantum state distinguishing, also called state discrimination, asks which of several known states has been supplied~\cite{Chefles00}. For two pure states with squared overlap $1-\varepsilon^2$, bounded-error discrimination requires $\Theta(\varepsilon^{-2})$ copies, whereas access to a preparation unitary and its inverse permits $O(\varepsilon^{-1})$ queries through amplitude amplification~\cite{Grover96,BHMT02}. Tang and Wright~\cite{TangWright25} show that generic amplitude amplification and estimation lose this quadratic advantage with forward-only access, using trace-estimation instances and compressed oracles.

\paragraph{Unitary discrimination.}
Ac{\'i}n~\cite{Acin01} and Duan, Feng, and Ying~\cite{DuanFengYing07} obtain perfect discrimination through repeated uses of fully specified unitary alternatives, whereas our promise specifies only one column. Forward-only access also permits unitary tomography to diamond-norm error $\varepsilon$ using $O(d^2/\varepsilon)$ queries~\cite{HaahEtAl23}, and universal unitary inversion~\cite{QuintinoEtAl19,ChenEtAl24}, including exact deterministic inversion with $O(d^2)$ queries \cite{ChenEtAl24}. Both unitary tomography \cite{HaahEtAl23} and unitary inversion \cite{odake,CYZ26} necessarily incur an $\Omega(d^2)$ overhead, making them unappealing in our high-dimensional regime.

\paragraph{Related lower bounds.}
Wang and Zhang~\cite{WangZhang25} relate sample and query complexity for state testing with block-encoding access. Tang, Wright, and Zhandry~\cite{TangWrightZhandry25} establish a corresponding simulation for state-preparation access, showing that forward and inverse queries offer at most a quadratic advantage over copies for decision problems. Our lower bound removes this advantage for the present promise when inverse queries are unavailable. Holman, Ramachandran, and Yirka~\cite{HolmanEtAl25} establish $\Theta(\sqrt{d})$ query complexity for permutation inversion with forward-only in-place access; this dimension-dependent hardness is useful in our construction. In a different learning setting, Cotler, Schuster, and Mohseni~\cite{CotlerEtAl23} obtain exponential separations between time-ordered and out-of-time-order correlators, likewise illustrating the value of reversed dynamics.

\paragraph{Quantum Zeno dynamics and learning.}
Our continuous-time algorithm uses quantum Zeno dynamics: repeated projections inhibit transitions~\cite{MisraSudarshan77} while permitting evolution within the projected subspace~\cite{FacchiPascazio02,ExnerIchinose05,FacchiPascazio08}. Related implementations use frequent pulses or strong couplings, connecting Zeno dynamics to dynamical decoupling~\cite{FacchiLidarPascazio04}. We isolate a known two-dimensional subspace that retains the distinguishing matrix element, charging total evolution time under the unknown Hamiltonian rather than the number or physical duration of interventions.

Franceschetto et al.~\cite{FranceschettoEtAl25} use a closely related confinement idea for Hamiltonian learning: frequent unitary kicks isolate small patches of a geometrically local Hamiltonian, whose interactions are reconstructed through process tomography. They provide finite-resource guarantees and a 109-qubit demonstration. Both approaches suppress unwanted couplings, but we use projective measurements and a promise on $H\ket{0}$ to obtain dimension-independent distinguishing without locality or sparsity assumptions. Their work does not address our discrete lower bound or the separation between oracle models.

\paragraph{Hamiltonian learning and testing.} 
Oracle learning in the Hamiltonian model has been studied in a number of other settings. Heisenberg-limited learning with forward evolution is known for local Hamiltonians through dynamical reshaping~\cite{HuangEtAl23}, and for sparse Hamiltonians without prescribed interaction structure~\cite{HuEtAl25}. Dutkiewicz, O'Brien, and Schuster~\cite{DutkiewiczEtAl24} establish Heisenberg-limited learning with control and $\Omega(\varepsilon^{-2})$ lower bounds without control for broad Hamiltonian classes. As is standard in this model, these results only allow forward time evolution under the Hamiltonian.

\paragraph{Continuous-time and discrete-time access.}
Continuous-time query algorithms were first proposed by Farhi and Gutmann~\cite{FarhiGutmann98}. Cleve et al.~\cite{CleveEtAl09} simulate continuous-time queries to a classical oracle with nearly the same number of discrete queries for structured input Hamiltonians and associated self-inverse phase oracles; however, they do not provide a simulation for arbitrary forward-only unitaries. In the intermediate case where there is a prescribed minimum evolution duration, Shin, Lee, and Oh~\cite{ShinEtAl26} obtain Heisenberg-limited learning under sparsity assumptions.

\subsection{Future Directions}

Our work leaves a number of future directions open. One direction is to improve the dimension dependence of \cref{thm:main-lower}, as the best known upper bound of $O(d^{3/2}\varepsilon^{-1})$ following from \cite{chen2025inversefreequantumstateestimation} leaves a significant gap. Alternatively, one could search for algorithms achieving $O(\varepsilon^{-1})$ scaling with milder dimension dependence, or achieving better tradeoffs in joint scaling with $\varepsilon$ and $d$.

Our algorithm establishing \cref{thm:main-upper} utilizes short evolution times interspersed with control operations. Prior work has studied an intermediate case where there is a prescribed minimum evolution duration, such as Shin, Lee, and Oh~\cite{ShinEtAl26} obtaining $\tilde O(\varepsilon^{-1})$ learning under sparsity assumptions. It would be interesting to determine how this would affect our algorithm, and in particular how long one can take the intermediate evolution times while still achieving $O(\varepsilon^{-1})$ scaling.

\section{The Algorithm}
\label{sec:algorithm}

In this section, we describe the algorithm solving \Cref{prob:cont-sd} with $O(\varepsilon^{-1})$ evolution time, establishing \Cref{thm:main-upper}. As a starting point, consider the following instance in two dimensions:
\begin{align}
    H_\mathrm{NO}&=\begin{pmatrix}
    0 & \varepsilon \\
    \varepsilon & 0
\end{pmatrix}, & H_\mathrm{YES}& =\begin{pmatrix}
    0 & 0 \\
    0 & 0
\end{pmatrix}.
\end{align}

Firstly, this example gives a lower bound of
$\Omega(\varepsilon^{-1})$ evolution time, since
\[
\norm{e^{-iH_\mathrm{NO}s}-e^{-iH_\mathrm{YES}s}} = \norm{e^{-iH_\mathrm{NO}s}- I}
\leq s\varepsilon
\]
for every $s\geq 0$. In particular, let $\ket{\psi^t_{\mathrm{NO}}}$ and $\ket{\psi^t_{\mathrm{YES}}}$ denote an algorithm's state after total evolution time $t$ under the two respective oracles. We may purify the algorithm, retaining all measurement outcomes in auxiliary registers. Oracle-independent unitaries preserve the distance between these states, while an oracle evolution of duration $s$ increases their distance by at most $s\varepsilon$, by the triangle inequality:
\begin{align*}
&\norm{
e^{-iH_\mathrm{NO}s}\ket{\psi^t_{\mathrm{NO}}}
-
e^{-iH_\mathrm{YES}s}\ket{\psi^t_{\mathrm{YES}}}
}\\
&\qquad\leq
\norm{\ket{\psi^t_{\mathrm{NO}}}-\ket{\psi^t_{\mathrm{YES}}}}
+
\norm{e^{-iH_\mathrm{NO}s}-e^{-iH_\mathrm{YES}s}}\\
&\qquad\leq
\norm{\ket{\psi^t_{\mathrm{NO}}}-\ket{\psi^t_{\mathrm{YES}}}}
+s\varepsilon.
\end{align*}
Here the evolutions act as the identity on auxiliary registers. Summing over all oracle evolutions therefore gives
\[
\norm{\ket{\psi^t_{\mathrm{NO}}}-\ket{\psi^t_{\mathrm{YES}}}}
\leq t\varepsilon.
\]
Since distinguishing with constant advantage requires constant
distance between the final states, any such algorithm requires
$t=\Omega(\varepsilon^{-1})$, establishing the lower bound.

This example will also illustrate how $O(\varepsilon^{-1})$ scaling is achievable.

\subsection{The Naive Algorithm}
\label{subsec:alg-naive}
A naive approach to state distinguishing using forward-only queries is to apply the operator $e^{-iHt}$ on an appropriate input state for a time $t=O(\varepsilon^{-1})$ and then measure. For the simple example instance $\{H_\mathrm{YES}, H_\mathrm{NO}\}$ described above, this algorithm works (and is in fact exactly optimal~\cite{Childs_2000}).

\begin{claim}
\label{claim:naive2d}
For the simple $\{H_\mathrm{YES}, H_\mathrm{NO}\}$ described above, the naive algorithm correctly distinguishes 
$H|0\rangle = 0$ from $H|0\rangle = \varepsilon|1\rangle$ using 
total evolution time $O(\varepsilon^{-1})$.
\end{claim}

\begin{proof}
In the YES case, $H_\mathrm{YES}\ket{0} = 0$, so the state remains in $\ket{0}$ over the whole evolution. In the NO case, $H\ket{0} = \varepsilon\ket{1}$, so the state rotates toward $\ket{1}$ at rate $\varepsilon$. In particular, we have \begin{align}
    e^{-iH_{\mathrm{NO}}t} \ket{0} &= \cos(\varepsilon t)\ket{0} - i\sin(\varepsilon t)\ket{1}.
\end{align}
After time $t =\frac{\pi}{2\varepsilon} =  O(\varepsilon^{-1})$ the two cases produce orthogonal states, which are perfectly distinguishable by a single measurement.
\end{proof}

Generalizations of this simple algorithm, such as unitary channel learning~\cite{HaahEtAl23}, solve the $2$-dimensional case in full generality. However, in the $d$-dimensional case such general algorithms scale with the dimension. If we apply the naive approach described, when $e^{-iHt}$ is applied to $\ket{0}$ the evolution is no longer confined to the plane spanned by $\ket{0}$ and $\ket{1}$.
The small difference $\varepsilon$ gets mixed with the components that have traveled outside the 2D subspace. This leads to dynamics in the NO case that need not coherently travel away from $\ket{0}$, instead meandering throughout the large number of dimensions. The distinguishing contributions coudl point in different directions and fail to reinforce one another, potentially leading to $O(\varepsilon^{-2})$ scaling. We will instead prevent the state from leaking into higher dimensions using the quantum Zeno effect.

\subsection{The General Algorithm}
\label{subsec:alg-general}

Let $\Pi$ denote the projector onto $\mathrm{span}\{|0\rangle, |1\rangle\}$ and let $\mathcal{M} = \{\Pi, \Pi^\perp\}$ be the corresponding projective measurement. We use repeated measurements to constrain the evolution to two dimensions, at which point we can apply the algorithm described in \Cref{subsec:alg-naive}.

\vspace{\baselineskip}

\begin{algorithm}[H]
\label{alg:main}
    \caption{Continuous-time state distinguishing}
    \begin{algorithmic}[1]
        \Procedure{Reduction to 2D}{}\newline
        \textbf{Input:} Time evolution oracle for $H$, target time $t$, state $\ket{\psi} \in \Pi$.\newline
        \textbf{Output:} A state arbitrarily close to $e^{-i\Pi H \Pi t} \ket{\psi}$.
        \State{\textbf{Choose} $n \in \mathbb N$ based on accuracy.}
        \State{\textbf{For} $j$ in $1,2,\dots, n$:}
        \State{\qquad \textbf{Apply} $e^{-iHt/n}$ to $\ket{\psi}$}
        \State{\qquad \textbf{Measure} $\ket{\psi}$ under $\mathcal M$.}
        \State {\textbf{Return} $\ket{\psi}$}
        \EndProcedure
    \end{algorithmic}
    \label{alg_1}
\end{algorithm}

\vspace{\baselineskip}

Algorithm~\ref{alg:main} converts forward access to the evolution oracle for $H$ into access to evolution under the projected Hamiltonian $\Pi H \Pi$, which acts nontrivially only on the 2D subspace. Using this reduction, we can distinguish the states by running the 2D learning algorithm discussed in \Cref{subsec:alg-naive} on the simulated evolution, and determining whether $\Pi H \Pi \ket{0} = 0$ or $\Pi H \Pi \ket{0} = \varepsilon\ket{1}$.

\section{Analysis}
\label{sec:analysis}

The goal of this section is to prove that Algorithm~\ref{alg:main} 
converges to ideal evolution under the projected Hamiltonian 
$\Pi H \Pi$. We first show that replacing the measurements with projections onto $\Pi$ is justified via the gentle measurement lemma, then show the convergence of the resulting projected evolution.

\subsection{Gentle Measurement}

Throughout we use $\Pi$ in two related ways: it denotes the orthogonal projector onto $\mathrm{span}\{\ket{0},\ket{1}\}$, and also the subspace $\mathrm{span}\{\ket{0},\ket{1}\}$ itself. Thus, $\ket{\psi}\in\Pi$ means that $\ket{\psi}$ lies in $\mathrm{span}\{\ket{0},\ket{1}\}$, or equivalently, $\Pi\ket{\psi}=\ket{\psi}$. The projector onto the orthogonal complement of this subspace is $\Pi^\perp=I-\Pi.$

\begin{lemma}[Gentle Measurement]
\label{lem:gentle}
Let $\ket{\psi} \in \Pi$ be a state, and let 
$\mathcal{M} = \{\Pi, \Pi^\perp\}$ be the projective measurement onto 
$\Pi$ and its complement. After applying $e^{-iHt/n}$ to $\ket{\psi}$, 
the probability that measurement under $\mathcal{M}$ returns 
$\Pi^\perp$ is at most $\|H\|^2t^2/n^2$.
\end{lemma}

\begin{proof}
Since $\ket{\psi} \in \Pi$, we have $\Pi^\perp \ket{\psi} = 0$. 
The probability of measuring $\Pi^\perp$ after applying $e^{-iHt/n}$ is
$$
\left\| \Pi^\perp e^{-iHt/n} \ket{\psi} \right\|^2 
= \left\| \Pi^\perp \left( e^{-iHt/n} - I \right) \ket{\psi} \right\|^2,
$$
where the equality uses $\Pi^\perp \ket{\psi} = 0$ to subtract 
$\Pi^\perp I \ket{\psi} = 0$. Since $\Pi^\perp$ is a projector, 
$\|\Pi^\perp\| \le 1$, so
$$
\left\| \Pi^\perp \left( e^{-iHt/n} - I \right) \ket{\psi} \right\| 
\le \left\| e^{-iHt/n} - I \right\| \leq \left\|H\right\|t/n.
$$
Squaring yields the bound $\|H\|^2t^2/n^2$.
\end{proof}

With this result in place, we can establish that the evolution converges to evolution governed by the projected Hamiltonian.

\subsection{Convergence to Projected Evolution}

By Lemma~\ref{lem:gentle}, a single measurement returns $\Pi^\perp$ with
probability at most $\|H\|^2t^2/n^2$. A run performs $n$ such measurements,
so the total probability of at least one failure (returning $\Pi^\perp$) is upper bounded by
$$
\Pr[\text{any }\Pi^\perp] \leq n \cdot \frac{\|H\|^2t^2}{n^2} \;=\; \frac{\|H\|^2t^2}{n},
$$
which disappears as $n \to \infty$. The algorithm's action on $\ket{\psi}$ is therefore approximated by
$\left( \Pi\, e^{-iHt/n} \right)^n$. We can now show that this is equivalent to evolution under the projected Hamiltonian $\Pi H \Pi$.

\begin{theorem}
\label{thm:main}
For $\ket{\psi} \in \Pi$,
$$
\lim_{n \to \infty} \left( \, e^{-iHt/n} \Pi\right)^n \ket{\psi} 
= e^{-i \Pi H \Pi t} \ket{\psi}.
$$
\end{theorem}
\begin{proof}
Since $\Pi\ket{\psi}=\ket{\psi}$, we can write
\[
\left(e^{-iHt/n}\Pi\right)^n\ket{\psi}
=
e^{-iHt/n}\left(\Pi e^{-iHt/n}\Pi\right)^{n-1}\ket{\psi}.
\]
Expanding the matrix exponential gives
\[
\Pi e^{-iHt/n}\Pi
=
\Pi+\frac{-i\Pi H\Pi t}{n}
+O\left(\frac{\|H\|^2t^2}{n^2}\right)
=
e^{-i(\Pi H\Pi)t/n}\Pi
+O\left(\frac{\|H\|^2t^2}{n^2}\right).
\]
Both $\Pi e^{-iHt/n}\Pi$ and
$e^{-i(\Pi H\Pi)t/n}\Pi$ have operator norm at most one.
Thus, the per-step error accumulates to at most
$O(\|H\|^2t^2/n)$ over $n-1$ steps. Consequently,
\[
\left(e^{-iHt/n}\Pi\right)^n\ket{\psi}
=
e^{-iHt/n}
e^{-i(\Pi H\Pi)t(n-1)/n}\ket{\psi}
+O\left(\frac{\|H\|^2t^2}{n}\right).
\]
For fixed $H$ and $t$, the first exponential converges to the
identity and the second converges to $e^{-i(\Pi H\Pi)t}$.
Taking $n\to\infty$ proves the claim.
\end{proof}

Theorem~\ref{thm:main} shows that the error of a single run of Algorithm~\ref{alg:main} can be made arbitrarily small, which can be used to simulate evolution under $\Pi H \Pi$ for some target time $t$ to arbitrary accuracy with no overhead. We can plug this into unitary channel learning algorithms, such as \cite[Theorem 1.1]{HaahEtAl23}, to learn whether $H\ket{0}=0$ or not in $O(\varepsilon^{-1})$ queries. In particular, this result allows learning the $2 \times 2$ submatrix of $e^{-i\Pi H\Pi t}$ acting on $\mathrm{span}\{\ket{0},\ket{1}\}$ (up to a global phase) to entry-wise precision $O(\varepsilon)$ using $O(\varepsilon^{-1})$ queries, and $\bra{1}e^{-i\Pi H\Pi}\ket{0} = 0$ in the YES and $|\bra{1}e^{-i\Pi H\Pi}\ket{0}| = \Omega(\varepsilon)$ in the NO case, as a result of the norm promise. Thus, evolution for time $\Theta(\varepsilon^{-1})$ suffices to distinguish the two cases.

\section{A Discrete Lower Bound}
\label{sec:lower}

The results above concern the continuous-time model, where the evolution oracle can be applied for any duration. The discrete query model instead treats a fixed unitary as one operation that cannot be subdivided. Algorithm~\ref{alg:main} does not apply directly to this model because its analysis requires the step duration $1/n$ to approach zero as $\varepsilon$ goes to zero. We now show that $\Theta(\varepsilon^{-2})$ forward-only queries are required in the discrete query model, proving \Cref{thm:main-lower}.

\subsection{The Hard Instance}

We construct the hard instance using a random permutation. Let $\pi$ be a permutation of $\{0, 1, \dots,\allowbreak d-1\}$ with $\pi(0) = 0$, and
let $P_\pi$ be the corresponding permutation operator,
$P_\pi \ket{m} = \ket{\pi(m)}$. Let $R_\varepsilon$ be a small rotation within $\mathrm{span}\{\ket{0}, \ket{1}\}$, in particular \begin{align}
    R_\varepsilon &\coloneqq \begin{pmatrix}
        \sqrt{1-\varepsilon^2} & -\varepsilon \\
        \varepsilon & \sqrt{1-\varepsilon^2}
    \end{pmatrix} \oplus I. \nonumber
\end{align}
The task is to distinguish
$$
U = P_\pi \qquad \text{from} \qquad V = R_\varepsilon P_\pi,
$$
for $\pi$ drawn uniformly at random subject to $\pi(0) = 0$, given forward
query access to one of them.

The condition $\pi(0) = 0$ means that $U\ket{0}=\ket{0}$ and $V\ket{0}=R_\varepsilon\ket{0}=\sqrt{1-\varepsilon^2}\ket{0}+\varepsilon\ket{1}$.
Thus, $U$ satisfies the YES promise and $V$ the NO promise exactly. In two dimensions, repeated application of $R_\varepsilon$ to
$\ket{0}$ rotates the state by $\approx \varepsilon$ per query, so if $\pi$ is the identity, the contributions add linearly and $O(\varepsilon^{-1})$ queries suffice. But for general $\pi$, composing with $P_\pi$ changes this. The amplitude placed on $\ket{1}$ by one query is sent by the next application of $P_\pi$ to $\ket{\pi(1)}$, which lies outside $\mathrm{span}\{\ket{0}, \ket{1}\}$ for all but one choice of $\pi(1)$. Contributions from different queries therefore go to different basis states. Their squared magnitudes add, so the total difference grows like $\varepsilon\sqrt{T}$ after $T$ queries, requiring $T = \Omega(\varepsilon^{-2})$.
\subsection{The Hybrid Argument}
\label{sec:discrete-hybrid}

We first quantify how much amplitude an algorithm can place on $\ket{\pi^{-1}(1)}$. The bound below follows from the bounded-error lower bound for inverting an in-place permutation from~\cite{FeffermanKimmel18}.

\begin{lemma}[Amplitude on $\pi^{-1}(1)$]
\label{lem:hidden-preimage-amplitude}
Let $\ket{\chi_\pi}$ be a state produced using
at most $t$ forward queries to $P_\pi$, with arbitrary
$\pi$-independent operations and auxiliary registers. For any
designated $d$-dimensional register $r$ of this state with computational basis
$\{\ket{0}_r,\ldots,\ket{d-1}_r\}$,
\[
\mathbb E_\pi
\left\|\left({}_r\!\bra{\pi^{-1}(1)}\otimes I\right)
\ket{\chi_\pi}\right\|^2
\leq O\left(\sqrt{\frac{t^2+1}{d}}\right),
\]
where the expectation is over uniform $\pi$ such that $\pi(0)=0$.
\end{lemma}

\begin{proof}
For the $t=0$ case, the bound is straightforwardly $O(1/d)$, which is subsumed by the $O(1/\sqrt{d})$ bound stated. Going forward, assume $t > 0$, so we can equivalently replace $t^2+1$ with $t^2$ inside the big $O$. Measuring the designated register of $\ket{\chi_\pi}$ gives the value $\pi^{-1}(1)$ with probability $\norm{\left(\bra{\pi^{-1}(1)}\otimes I\right)
\ket{\chi_\pi}}^2$.
Let $p$ be the average of this probability over $\pi$.
The in-place permutation adversary bound of
Fefferman and Kimmel~\cite[Lemma~2]{FeffermanKimmel18} implies that any algorithm finding $\pi^{-1}(1)$ with constant probability for every $\pi$ requires at least $\Omega(\sqrt{d})$ queries, by viewing $\pi$ as a uniform permutation on $\ket{1}, \dots, \ket{d-1}$. Permutation inversion is worst-case to average-case reducible by composing $\pi$ with a random permutation, so this bound applies in the average case as well.

If the adversary achieved a probability $p = \omega(\sqrt{{t^2}/{d}})$ of recovering $\pi^{-1}(1)$, then repeating for $o(\sqrt{{d}/{t^2}})$ repetitions and re-randomizing the permutation each time would yield a bounded error algorithm that inverts with $t \cdot o(\sqrt{{d}/{t^2}}) = o(\sqrt{d})$ queries, a contradiction.
\end{proof}

Note that the bound in the prior lemma could potentially be tightened to $O(\frac{t^2+1}{d})$. This bound would follow from replacing classical repetition with amplitude amplification, but naive amplitude amplification does not work as the the inverse oracle is not available. An alternative proof could show the improved bound, though the loose bound will suffice for our application.

Given that no algorithm can place significant mass on $\pi^{-1}(1)$, we can bound the distance increase between the potential states of the algorithm from one application of the oracle, as captured by the following lemma.
The lemma allows the states to depend on $\pi$, assuming only a bound on the amplitude placed on $\pi^{-1}(1)$.

\begin{lemma}[One-query progress]
\label{lem:singlequery-progress}
Let $\ket\psi$ and $\ket\phi$ be normalized states of the
query register and an arbitrary workspace. Write
\[
\ket\psi=\sum_{x=0}^{d-1}\ket{x}\ket{a_x},
\qquad
\ket\phi=\sum_{x=0}^{d-1}\ket{x}\ket{b_x},
\]
where the workspace vectors need not be normalized. Then
\begin{align}
\left\|(P_\pi\otimes I)\ket\psi
 -(R_\varepsilon P_\pi\otimes I)\ket\phi\right\|^2
&\leq {}\|\ket\psi-\ket\phi\|^2+2(1-\sqrt{1-\varepsilon^2})
\nonumber\\
&\quad+2\varepsilon\bigl(\|\ket{a_{\pi^{-1}(1)}}\|+\|\ket{b_{\pi^{-1}(1)}}\|\bigr) 
\label{eq:singlequery-progress}
\end{align}
\end{lemma}

\begin{proof}
Expanding the squared distance gives the exact identity
\begin{align*}
&\left\|(P_\pi\otimes I)\ket\psi
 -(R_\varepsilon P_\pi\otimes I)\ket\phi\right\|^2
 -\|\ket\psi-\ket\phi\|^2\\
&\quad=2\operatorname{Re}\left(
 \langle\psi|\phi\rangle
 -\langle\psi|(P_\pi^\dagger R_\varepsilon P_\pi\otimes I)|\phi\rangle
 \right)\\
&\quad=2\operatorname{Re}\Bigl(
 \langle a_0|b_0\rangle
 +\langle a_{\pi^{-1}(1)}|b_{\pi^{-1}(1)}\rangle\\
&\qquad\qquad
 -\langle a_0|\bigl(\sqrt{1-\varepsilon^2}\ket{b_0}
                    -\varepsilon\ket{b_{\pi^{-1}(1)}}\bigr)
 -\langle a_{\pi^{-1}(1)}|
       \bigl(\varepsilon\ket{b_0}
             +\sqrt{1-\varepsilon^2}\ket{b_{\pi^{-1}(1)}}\bigr)
 \Bigr)\\
&\quad=2(1-\sqrt{1-\varepsilon^2})\operatorname{Re}
 \bigl(\langle a_0|b_0\rangle+\langle a_{\pi^{-1}(1)}|b_{\pi^{-1}(1)}\rangle\bigr)
 +2\varepsilon\operatorname{Re}
 \bigl(\langle a_0|b_{\pi^{-1}(1)}\rangle-\langle a_{\pi^{-1}(1)}|b_0\rangle\bigr).
\end{align*}
The first equality uses unitarity, and the second cancels all components outside the two coordinates affected by \(R_\varepsilon\).
By Cauchy--Schwarz and normalization, the first real part is
at most $1$, while the second is at most
$\|\ket{a_0}\|\|\ket{b_{\pi^{-1}(1)}}\|+\|\ket{a_{\pi^{-1}(1)}}\|\|\ket{b_0}\|
\leq\|\ket{a_{\pi^{-1}(1)}}\|+\|\ket{b_{\pi^{-1}(1)}}\|$.
\end{proof}

With these lemmas in place, we can prove the complete lower bound.

\begin{theorem}[Discrete distinguishing lower bound]
\label{thm:discrete-distinguishing-lower}
Let $d\geq3$ and $0<\varepsilon\leq1/2$. For any $T$-query
forward-only algorithm, let $\rho_U$ and $\rho_V$ be its final
states averaged over $\pi$ in the two oracle cases above. Then
\begin{equation}
\frac12\|\rho_U-\rho_V\|_1
\leq
O\left(\varepsilon \sqrt{T}
       +\frac{\varepsilon^{1/2} T^{3/4}}{d^{1/8}}\right).
\label{eq:finite-dimensional-distinguishing}
\end{equation}
Consequently, distinguishing the two promises with success
probability at least $2/3$ requires
\[
T=\Omega\left(
\min\left\{\varepsilon^{-2},
             d^{1/6}\varepsilon^{-2/3}\right\}
\right).
\]
If $d\geq\varepsilon^{-8}$, the query
complexity is $\Theta(\varepsilon^{-2})$.
\end{theorem}

\begin{proof}
Let $\ket{\psi_t^\pi}$ and
$\ket{\phi_t^\pi}$ denote the algorithm's states after $t$ queries in
the $U$ and $V$ executions, respectively, under permutation $\pi$. Let
\[
D_t^2=\mathbb E_\pi
\|\ket{\psi_t^\pi}-\ket{\phi_t^\pi}\|^2.
\]
The initial states coincide, so $D_0=0$. Oracle-independent unitaries between queries preserve each pairwise distance, so it suffices to analyze the queries.

Immediately before query $t+1$, apply \cref{lem:singlequery-progress} to the two states. Each has been generated by $t$ forward permutation queries: in the NO case corresponding to oracle $V$, one can replace each $V$ call by $P_\pi$ followed by the known rotation $R_\varepsilon$, to view the algorithm as querying $P_\pi$ rather than $V$. Thus \cref{lem:hidden-preimage-amplitude} applies to each execution, which combined with Cauchy-Schwarz gives
\[
\mathbb E_\pi\|a_{\pi^{-1}(1)}\|,
\ \mathbb E_\pi\|b_{\pi^{-1}(1)}\|
\leq O\left(\frac{\sqrt{t+1}}{d^{1/4}}\right).
\]
Taking expectations in \eqref{eq:singlequery-progress}  yields
\[
D_{t+1}^2\leq D_t^2+2(1-\sqrt{1-\varepsilon^2})
              +O\left(\frac{\varepsilon\sqrt{t}}{d^{1/4}}\right).
\]
Summing over $t=0,\ldots,T-1$, and using
$1-\sqrt{1-\varepsilon^2}=\varepsilon^2/(1+\sqrt{1-\varepsilon^2})\leq\varepsilon^2$, gives
\begin{align*}
D_T^2
&\leq2T(1-\sqrt{1-\varepsilon^2})+O\left(\frac{\varepsilon}{d^{1/4}}
                  \sum_{s=1}^T\sqrt{s}\right)\\
&\leq2T\varepsilon^2
       +O\left(\frac{\varepsilon T^{3/2}}{d^{1/4}}\right).
\end{align*}
Convexity of trace distance, its pure-state bound by Euclidean
distance, and Cauchy--Schwarz imply
\[
\frac12\|\rho_U-\rho_V\|_1
\leq\mathbb E_\pi
\|\ket{\psi_T^\pi}-\ket{\phi_T^\pi}\|
\leq D_T.
\]
Success probability $2/3$ requires trace distance at least $1/3$, establishing the claim.
\end{proof}

We do not characterize the optimal dimension dependence of this bound, and it is likely that it can be improved.

\section*{AI Statement}

All of the main results and definitions are due to the authors. GPT-6-Astra was used to check and simplify calculations, help prepare the bibliography, and find related works.

\section*{Acknowledgments}

This work received support from the National Science Foundation (grant 26-17356) and the Department of Energy (grant DE-SC0020264 and the Office of Science, Office of Advanced Scientific Computing Research, Accelerated Research in Quantum Computing program).

\printbibliography
\end{document}